\documentclass[11pt]{article}
\pdfoutput=1

\usepackage{amsfonts}
\usepackage{amssymb}
\usepackage{amstext}
\usepackage{amsmath}
\usepackage{enumerate}
\usepackage{algorithm}
\usepackage{algorithmic}

\usepackage{graphicx}

\usepackage{bbm}

\usepackage{amsthm}

\usepackage{thmtools}

\usepackage{verbatim} % comment environment

\usepackage[usenames,dvipsnames,svgnames,table]{xcolor}
\definecolor{dartmouthgreen}{rgb}{0.05, 0.5, 0.06}
\definecolor{ceruleanblue}{rgb}{0.16, 0.32, 0.75}

\usepackage[colorlinks=true,pdfpagemode=UseNone,citecolor=dartmouthgreen,linkcolor=ceruleanblue,urlcolor=BrickRed,pdfstartview=FitH]{hyperref}
\usepackage{cleveref}

\usepackage{framed}
\usepackage{enumitem}


\newtheorem{theorem}{Theorem}[section]

\newtheorem{lemma}[theorem]{Lemma}
\newtheorem{definition}[theorem]{Definition}

\newtheorem{proposition}[theorem]{Proposition}
\newtheorem{claim}[theorem]{Claim}

\newtheorem*{problem*}{Problem}

\newtheorem*{remark*}{Remark}

\numberwithin{equation}{section}
\numberwithin{table}{section}

\newcommand{\poly}{\operatorname{poly}}

\newcommand{\junk}[1]{}

\newcommand{\TV}{{\rm TV}}

\newcommand{\norm}[1]{\left\lVert#1\right\rVert}

\newcommand{\vertiii}[1]{{\left\vert\kern-0.25ex\left\vert\kern-0.25ex\left\vert #1 \right\vert\kern-0.25ex\right\vert\kern-0.25ex\right\vert}}

\def\b1{{\bf 1}}
\def\eps{{\epsilon}}

\title{A Combinatorial Characterization of Constant Mixing Time}

\author{
 Lap Chi Lau\footnote{Cheriton School of Computer Science, University of Waterloo. Supported by an NSERC Discovery Grant. 
 },~~~~~
 Raymond Liu\footnote{Cheriton School of Computer Science, University of Waterloo. Supported by an NSERC Discovery Grant. 
 }}

\date{}

\begin{document}

\begin{titlepage}
\def\thepage{}
\thispagestyle{empty}

\maketitle

\begin{abstract}
Classical spectral graph theory characterizes graphs with logarithmic mixing time.
In this work, we present a combinatorial characterization of graphs with constant mixing time.
The combinatorial characterization is based on the small-set bipartite density condition,
which is weaker than having near-optimal spectral radius and is stronger than having near-optimal small-set vertex expansion.
\end{abstract}

\end{titlepage}

\thispagestyle{empty}

%\tableofcontents

\newpage

\section{Introduction}

We start with a brief review of some background in random walks and spectral graph theory.
Let $G=(V,E)$ be a $d$-regular graph and $n := |V|$.
Let $P$ be the random walk matrix of $G$, with stationary distribution $\pi=\vec{1}/n$ as $G$ is regular.  
The $\eps$-mixing time of the random walks is defined as
\[
\tau_\eps(P) := \min \Big\{ t ~\Big|~ \frac12\norm{P^t p_0 - \pi}_1 \leq \eps \textrm{ for any initial distribution } p_0\Big\}.
\]
Let $1 = \lambda_1 \geq \lambda_2 \geq \cdots \geq \lambda_n \geq -1$ be the eigenvalues of $P$ and $\lambda := \max\{\lambda_2,|\lambda_n|\}$ be the spectral radius of $P$.
The graph $G$ is called a spectral expander if $\lambda$ is a constant strictly smaller than one.
Standard spectral analysis~\cite{LPW09} shows that 
\begin{equation} \label{e:TV}
 \frac12\norm{P^t p_0 - \pi}_1 \leq \lambda^t \cdot \norm{p_0}_2 \cdot \sqrt{n}.
\end{equation}
This implies that the mixing time of random walks on a spectral expander is $O(\log n)$,
and this upper bound is optimal as the diameter of the graph is $\Omega(\log n)$ when $d$ is a constant.
Cheeger's inequality states that $G$ is a spectral expander if and only if $G$ is a combinatorial expander (i.e., with constant edge conductance); see \autoref{s:prelim}.
This gives a combinatorial characterization of graphs with $O(\log n)$ mixing time.
The relations between eigenvalues, combinatorial expansion, and mixing time are fundamental results in spectral graph theory.

Inspired by the recent development in constant-hop expander graphs (see \cite{HRG22,HHG25} and the references therein), 
where the focus is on sending multicommodity flows using paths of {\em constant} length, 
we are interested in characterizing graphs with constant mixing time, as these form a nice class of constant-hop expander graphs.
For a $d$-regular graph to have constant diameter, a necessary condition is that $d \geq n^\xi$ for some small constant $\xi > 0$, so we focus on graphs in this moderately-dense regime as in~\cite{HRG22,HHG25}. 
Even in this regime, it is not difficult to construct spectral expanders with $\Omega(\log n)$ diameter, so we need to look for stronger conditions to guarantee constant mixing time.
We say a $d$-regular graph has an {\em inverse-polynomial spectral radius} if $\lambda \lesssim 1/d^c$ for some constant $c \in (0,\frac12]$.  
From \eqref{e:TV}, observe that $n^\xi$-regular graphs with $\lambda \lesssim 1/d^c$ have constant mixing time $O(1/(c \xi))$. 
Graphs with inverse-polynomial spectral radius also exhibit stronger combinatorial expansion properties:
Tanner's theorem implies that such graphs have vertex expansion $\Omega(d^{2c})$ for sets of size $O(n/d^{2c})$; see \autoref{s:prelim}.
Graphs in these regimes satisfy interesting properties but were not explored much before.

These lead us to study the relations between inverse-polynomial spectral radius, small-set vertex expansion, and constant mixing time.
One natural question is whether an $n^\xi$-regular graph with near-optimal small-set vertex-expansion has constant mixing time.
In \autoref{s:SSE}, we provide a negative example to this question, which suggests that an even stronger combinatorial condition is required to guarantee constant mixing time.

Another natural question is whether inverse-polynomial spectral radius is necessary to guarantee constant mixing time.
To answer this question, 
we consider a combinatorial characterization of the spectral radius through the expander mixing lemma~\cite{AC88}:
If a $d$-regular graph $G=(V,E)$ has spectral radius $\lambda$, 
then
\begin{equation} \label{e:EML}
\bigg| \big|E(S,T)\big| - \frac{d|S||T|}{n} \bigg| \leq \lambda d \sqrt{|S||T|}
\quad \textrm{ for any } S, T \subseteq V,
\end{equation}
where $E(S,T) := \{(u,v) \mid u \in S, v \in T, uv \in E\}$ is the set of ordered edges where $u \in S$ and $v \in T$.
The converse of the expander mixing lemma by Bilu and Linial~\cite{BL06} shows that if \eqref{e:EML} is satisfied for all disjoint $S, T \subseteq V$,
then the graph has spectral radius $O(\lambda \cdot \log(1 + \frac{1}{\lambda}))$.
Thus, if \eqref{e:EML} is satisfied for $\lambda \leq O(1/d^c)$ for some constant $c \in (0, \frac12]$ and $d \geq n^\xi$ for some constant $\xi > 0$, then the graph has constant mixing time.
This provides a combinatorial sufficient condition for constant mixing time, but we will show that it is not a necessary condition.

\subsection{Our Results}

We show that only the upper bounds in \eqref{e:EML} are needed for constant mixing time, 
simplifying \eqref{e:EML} to a condition about the bipartite density between two sets.

\begin{definition}[$\alpha$-Bipartite Density] \label{d:bipartite-density}
Let $G=(V,E)$ be a $d$-regular graph.
For any $\alpha\in (\sqrt d,d]$, we say that $G$ satisfies the $\alpha$-bipartite density condition if
\begin{equation}\label{e:bipartite-density}
     |E(S,T)| \leq \frac{d|S||T|}{n}+\alpha\sqrt{|S||T|} \quad \textrm{ for all } S,T \subseteq V.
\end{equation}
If \eqref{e:bipartite-density} is only satisfied for sets $S,T$ with $|S|,|T| \leq \delta n$ for some $\delta\leq 1$, then we say that $G$ satisfies the $\delta$-small-set $\alpha$-bipartite density condition.
\end{definition}

We note that the $\alpha$-bipartite density condition is weaker than having spectral radius $\alpha/d$ by \eqref{e:EML}, but it is stronger than having near-optimal small-set vertex expansion (see \autoref{s:SSE}).

Our main results are that this condition implies fast mixing time.

\begin{theorem}[Upper Bounding Mixing Time by Bipartite Density] \label{t:upper-bound}
Let $G=(V,E)$ be a $d$-regular graph with $d = n^\xi$ for some constant $\xi > 0$.
If $G$ satisfies the $\alpha$-bipartite density condition for $\alpha\lesssim d/(\log d)^2$, then
\[
\tau_{1/n}(P) \lesssim \Big(\frac{\log n}{\log (d/\alpha)}\Big)^2.
\]
In particular, if $\alpha=d^{1-c}$ for some constant $c\in (0,\frac12]$, 
this implies constant mixing time such that
\[
\tau_{1/n}(P)\lesssim \frac{1}{c^2 \xi^2}.
\]
\end{theorem}

The standard definition of mixing time is $\tau_{1/3}(P)$, and a well-known fact~\cite{LPW09} is that $\tau_{1/n}(P) \lesssim \tau_{1/3}(P) \cdot \log n$.
In \autoref{t:upper-bound}, we bound $\tau_{1/n}(P)$ directly without losing the logarithmic factor, matching the constant mixing time result for graphs with inverse-polynomial spectral radius.

Furthermore, we show an improved upper bound on standard mixing time using only the small-set bipartite density condition.

\begin{theorem}[Upper Bounding Mixing Time by Small-Set Bipartite Density] \label{t:weak-mixing-time}
Let $G=(V,E)$ be a $d$-regular graph with $d = n^\xi$ for some constant $\xi > 0$.
If $G$ satisfies the $\delta$-small-set $\alpha$-bipartite density condition for some $\alpha\lesssim d/(\log d)^2$ and $\delta \gtrsim \alpha/d$, then
\[
\tau_{1/3}(P)\lesssim \frac{\log n}{\log(d/\alpha)}.
\]
In particular, if $\alpha=d^{1-c}$ for some constant $c\in (0,\frac12]$
and \eqref{e:bipartite-density} holds for all sets $S,T \subseteq V$ with $|S|,|T| \lesssim n/d^c$, then this implies constant standard mixing time such that
\[
\tau_{1/3}(P)\lesssim \frac{1}{c \xi}.
\]
\end{theorem}

We remark that constant standard mixing time implies the graph is a constant-hop expander graph. 
It is our hope that the small-set bipartite density condition can lead to a simpler cut-matching game for constructing constant-hop expanders.

We also establish a lower bound on the mixing time using $\delta$-small-set $\alpha$-bipartite density. 
The following theorem states that the existence of a dense bipartite structure between two small sets implies slow mixing time.

\begin{theorem}[Lower Bounding Mixing Time by Bipartite Density] \label{t:lower-bound}
Let $G$ be a $d$-regular graph. 
If there exist $S,T\subseteq V$ such that 
\[
|E(S,T)|\geq \frac{d|S||T|}{n}+\alpha\sqrt{|S||T|}
\quad \textrm{and} \quad
|S|,|T|\leq \delta n,
\] 
then
\[
\tau_{1/n}(P)\gtrsim \frac{\log(1/\delta)}{\log(d/\alpha)}.
\]
In particular, if there are two small sets with high bipartite density such that $|S|,|T| \leq n^{1-\eps}$ for some constant $\eps$ and $\alpha = \Omega(d/(\poly\log d))$, 
then the graph has non-constant mixing time such that $\tau_{1/n}(P) \gtrsim (\log n)/\log \log d$.
\end{theorem}

To summarize,
we can view the $\delta$-small-set $\alpha$-bipartite density condition as a loose characterization of constant mixing time: 
If all sets of large enough size have low bipartite density, then the graph has constant mixing time; if some sets of small enough size have high bipartite density, then the graph has non-constant mixing time.

We think the proof approach in \autoref{t:upper-bound} and \autoref{t:weak-mixing-time} is also interesting that it provides a clean and direct way to upper bound the mixing time using a combinatorial condition, without going through a spectral argument as usual.

\section{Preliminaries}\label{s:prelim}

We write $f \lesssim g$ if $f = O(g)$, $f \gtrsim g$ if $f = \Omega(g)$, and $f \asymp g$ if $f = \Theta(g)$.

We assume the given graph $G=(V,E)$ is a $d$-regular graph throughout this paper,
with $n:=|V|$ vertices and $m:=|E|$ edges.
Let $A$ be the adjacency matrix of $G$,
and let $P:=A/d$ be the normalized adjacency matrix,
which is also the transition matrix of random walks, as $G$ is $d$-regular.

Let $1 = \lambda_1 \geq \lambda_2 \geq \cdots \geq \lambda_n \geq -1$ be the eigenvalues of $P$.
We call $\lambda := \max\{\lambda_2,|\lambda_n|\}$ the spectral radius of $G$.
A well-known result by Alon and Boppana~\cite{Alon-Boppana} establishes that $\lambda \geq 2\sqrt{d-1}/d$ as $n \to \infty$.
A graph is called Ramanujan~\cite{LPS} if $\lambda \leq 2\sqrt{d-1}/d$.

\subsubsection*{Variation Distance and Mixing Time}

The stationary distribution $\pi$ of the transition matrix $P=A/d$ is the uniform distribution $\vec{1}/n$.
The variation distance between any two probability distributions $p, q$ is defined as $d_{\TV}(p, q)=\frac12\norm{p-q}_1$. The variation distance at step $t$ of the random walk is defined as
\[
d_{\TV}(t)=\max_pd_{\TV}(P^tp,\pi),
\]
where the maximum is over all initial probability distributions $p$ on $V$. 
Given $\varepsilon>0$, the $\varepsilon$-mixing time of the random walk is defined as
\[
\tau_{\varepsilon}(P)=\min\{t\mid d_{\TV}(t)\leq \varepsilon\}.
\]

Standard spectral analysis in \eqref{e:TV} yields
$d_{TV}(P^tp,\pi) \lesssim \lambda^t \cdot \norm{p}_2 \cdot \sqrt{n}$.
It follows that the mixing time is upper bounded by $O(\log(n)/(1-\lambda))$,
and so when $G$ is a spectral expander, the mixing time is bounded by $O(\log n)$.

A standard fact in Markov chain \cite{LPW09} shows that for any $k\in\mathbb N$,
\begin{equation}\label{e:standard}
    d_{\TV}(kt)\leq (2d_{\TV}(t))^k.
\end{equation}

In particular, any graph with mixing time $\tau_{1/3}(P)\lesssim\log n$ implies $\tau_{1/n}(P) \lesssim (\log n)^2$.

\subsubsection*{Edge Conductance}

Given an undirected graph $G=(V,E)$ and $S,T \subseteq V$,
define 
\[
E(S,T) := \{(u,v) \mid u \in S, v \in T, uv \in E\},
\]
where an edge with $u,v \in S \cap T$ is counted twice,
as both $(u,v)$ and $(v,u)$ are in $E(S,T)$. The edge boundary of $S$ is defined as $\delta(S):=E(S,V\setminus S)$.

The second eigenvalue $\lambda_2$ is closely related to the edge conductance of the graph,
defined as
\[
\phi(G) = \min_{S \subseteq V : |S| \leq n/2} \frac{|\delta(S)|}{d |S|}.
\]
Cheeger's inequality~\cite{Cheeger, MR782626, MR875835} states that

\begin{equation}\label{e:cheeger}
    \frac{1}{2} ( 1- \lambda_2 ) \leq \phi(G) \leq \sqrt{2(1-\lambda_2)}.
\end{equation}
The edge conductance characterizes the mixing time of constant degree graphs: 
\[
\frac{1}{\phi(G)}\lesssim \tau_{1/3}(W)\lesssim \frac{\log n}{\phi(G)^2},
\]
where $W=\frac12(I+P)$ is the lazy random walk. We only consider non-lazy random walk $P$ in this paper. Since for any initial distribution of the form $p=\chi_v$, where $\chi_v$ is the indicator vector of a vertex $v\in V$, a staying probability of $\frac12$ automatically implies $\tau_{1/n}(P)$ is non-constant.

\subsubsection*{Graphs with Inverse-Polynomial Spectral Radius}

A graph $G$ exhibits stronger probabilistic and combinatorial properties when it has an inverse-polynomial spectral radius such that $\lambda \lesssim 1/d^c$ for some constant $c \in (0,1/2]$.

For random walks, it follows from \eqref{e:TV} that if $d=n^\xi$ for some constant $\xi > 0$, then the mixing time of a graph with inverse-polynomial spectral radius is a constant, upper bounded by $O(1/(c\xi))$.

A graph with inverse-polynomial spectral radius also has large small-set vertex expansion.
Define the vertex expansion of a set $S \subseteq V$ as
\[
\psi(S) := \frac{|\partial(S)|}{|S|}
\quad \textrm{ where }
\partial(S) := \{ v \notin S \mid \exists u \in S \textrm{ with } uv \in E\}.
\]
Tanner~\cite{Tanner} proved that
\begin{equation} \label{e:Tanner}
\psi(S) \geq \Big( \frac{|S|}{n} \big( 1 - \lambda^2 \big) + \lambda^2 \Big)^{-1} - 1.
\end{equation}
In particular, when $G$ is Ramanujan, it has near-optimal small-set vertex expansion such that $\psi(S) = \Omega(d)$ for sets of size up to $\Omega(n/d)$.

We remark that the small-set vertex expansion condition can also be derived from the $\alpha$-bipartite density condition with $\alpha=\lambda d$ (instead of the spectral radius $\lambda$).

\section{Upper Bounds on Mixing Time}\label{s:upper-bound}

In this section, we prove \autoref{t:upper-bound} and \autoref{t:weak-mixing-time}.

We show that for a graph that satisfies the small-set bipartite density condition,
the variation distance to the stationary distribution after $O(1/(c\xi))$ steps of random walks is essentially upper bounded by how close $\delta$ is to 1.

\begin{theorem}[Upper Bounding Variation Distance by Small-Set Bipartite Density]\label{t:VD}
Let $G$ be a $d$-regular graph with $d = n^\xi$ for some constant $\xi > 0$. If $G$ satisfies the $\delta$-small-set $\alpha$-bipartite density condition for some $\alpha\lesssim d/(\log d)^2$, then 
\[
d_{\TV}(t)\leq \frac14\cdot \sqrt{1-\delta} + O\Big(\frac{1}{n^{c\xi/8}}\Big)
\quad {\rm for} \quad
t\geq \frac{4}{c\xi}+1.
\]
where $c:=1-\log_d\alpha$ such that $\alpha=d^{1-c}$.
\end{theorem}

Assuming \autoref{t:VD}, the proofs of \autoref{t:upper-bound} and \autoref{t:weak-mixing-time} follow easily.

\begin{proof}[Proof of \autoref{t:weak-mixing-time}]
Since $c=1-\log_d\alpha$ and $n=d^{1/\xi}$, it follows that 
\begin{equation}\label{e:n^{c xi/8}}
c\xi = \xi - \xi \cdot \frac{\log\alpha}{\log d} = \xi-\frac{\log\alpha}{\log n}
\quad \implies \quad
\frac{1}{n^{c\xi}}=\frac{1}{n^{(\xi-\log_n\alpha)}}=\frac{\alpha}{d}.
\end{equation}
By \eqref{e:n^{c xi/8}} and \autoref{t:VD}, for $t\geq 4/c\xi +1$,
\[
d_{\TV}(t)
\leq \frac14\cdot \sqrt{1-\delta}+O\Big(\Big(\frac{\alpha}{d}\Big)^{\frac18}\Big)
\leq \frac14+O\Big((\log d)^{-\frac14}\Big)\leq \frac13,
\]
where the second inequality follows by the assumption that $\alpha\lesssim d/(\log d)^2$, and the last inequality holds for sufficiently large $d$. 
This implies that 
$$\tau_{1/3}(P) \lesssim \frac{1}{c\xi}=\frac{\log d}{\xi\log(d/\alpha)}=\frac{\log n}{\log (d/\alpha)}$$
\end{proof}

To prove \autoref{t:upper-bound}, we apply the standard fact in \eqref{e:standard}.

\begin{proof}[Proof of \autoref{t:upper-bound}]
Here we assume $\alpha$-bipartite density so $\delta=1$. 
For $t\geq 4/c\xi+1$, as in \eqref{e:n^{c xi/8}}, 
\[
d_{\TV}(t)
\leq \frac14 \cdot \sqrt{1-\delta} + O\Big(\frac{1}{n^{c\xi/8}}\Big)
=O\Big( \frac{\alpha}{d} \Big)^{\frac18}:=\beta,
\]
By \eqref{e:standard}, 
for $k\geq (\log \frac1n)/\log (2\beta)$, 
\[
d_{\TV}(kt) 
\leq (2 \cdot d_{\TV}(t))^k 
\leq (2\beta)^k \leq \frac1n.
\]
Since $c=1-\log_d\alpha = \log(d/\alpha)/\log d$, we conclude that 
\[
\tau_{1/n}(P)
\lesssim \frac{\log n}{\log (d/\alpha)}\cdot \Big(\dfrac{4}{c\xi}+1\Big)
\lesssim \frac{\log n}{c\xi\log(d/\alpha)}
= \frac{\log d}{\log(d/\alpha)}\cdot \frac{\log n}{\xi\log(d/\alpha)}
= \Big( \frac{\log n}{\log(d/\alpha)}\Big)^2.
\]
\end{proof}

\subsection{Bounding Mixing Time via 2-Norm}

The main idea is to measure the mixing progress by the $\ell_2$-norm of the random walk distribution. 
For any distribution $p$, $\|p\|^2_2$ is lower bounded by that of the stationary distribution, which is $1/n$. 
The following proposition measures the progress of $\|p\|^2_2$ approaching $C_\delta/n$, where $C_\delta$ is a constant depending on how close $\delta$ is to 1.

\begin{proposition}[$2$-Norm Decrease] \label{p:2-Norm Bound}
Let $G$ be a $d$-regular graph with $d = n^\xi$ for some constant $\xi > 0$.
If $G$ satisfies the $\delta$-small set $\alpha$-bipartite density condition with $\alpha \log d / (\xi d) \leq 1$ and $\delta \geq \frac{\sqrt 5-2}{2}\cdot \alpha/d $, then  
\[
\|Pp\|_2^2\leq \frac{C_\delta}{n}+O\bigg(\sqrt{\frac{\alpha\log d}{\xi d}}\bigg)\cdot \|p\|_2^2
\]
for any probability distribution $p$, where $C_\delta=\delta+5(1-\delta)/4$.
\end{proposition}

\begin{proof}
Let $p$ be an arbitrary probability distribution and $Pp$ be the distribution after one step of random walks. 
Assume $p(1)\geq\cdots\geq p(n)$ without loss of generality, and let $\sigma$ be a permutation of $[n]$ such that $Pp(\sigma(1))\geq Pp(\sigma(2))\geq \cdots \geq Pp(\sigma(n))$. Let $T=[\sigma(1),\cdots, \sigma(k)]$ be the largest $k$ entries in $Pp$ for some $k\in [n]$, and let $d_T(i)=|\{j\in T\mid ij\in E\}|$ denotes the number of vertices in $T$ adjacent to some vertex $i\in V$.
Note that
\begin{equation}\label{e:Pp1tok}
\sum_{1\leq j\leq k}Pp(\sigma(j))
=\sum_{1\leq j\leq k} ~\sum_{i:(i,\sigma(j))\in E}\frac{p(i)}{d}
=\sum^{n}_{i=1}\frac{p(i)}{d}\cdot d_T(i)
\leq \sum^n_{i=1}\frac{p(i)}{d}\cdot d'_T(i),
\end{equation}
where $d_T'(1)\geq \cdots\geq d_T'(n)$ is a permutation of $d_T(1),\cdots, d_T(n)$ sorted from largest to smallest, with the last inequality follows from rearrangement inequality. 

Consider first the case that $|T|=k\leq \delta n$. For any $r\in \{1,\cdots, \delta n\}$, it follows from the $\delta$-small-set $\alpha$-bipartite density condition and an averaging argument that
\begin{equation}\label{e:d'_T(i)}
\sum_{i\leq r}d_T'(i)
\leq \frac{dr|T|}{n}+\alpha\sqrt{r |T|}
\quad \implies \quad 
d_T'(r)\leq \frac{d|T|}{n}+\alpha\sqrt{\frac{|T|}{r}}.
\end{equation}
And, for any $r\geq \delta n$, it follows from the sortedness of $d_T'(i)$ that
\begin{equation}\label{e:d'_T(delta n)}
d_T'(r)\leq d_T'(\delta n)\leq \frac{d|T|}{n}+\alpha\sqrt{\frac{|T|}{\delta n}}.
\end{equation}
Denote $\bar\alpha := \alpha/d$. 
Combining \eqref{e:d'_T(i)}, \eqref{e:d'_T(delta n)} with \eqref{e:Pp1tok}, it follows that
\begin{align*}
\sum_{1\leq j\leq k}Pp(\sigma(j))
\leq \sum^{n}_{i=1}\frac{p(i)}{d}\cdot d_T'(i)
&\leq \sum^{\delta n}_{i=1}\frac{p(i)}{d}\cdot \bigg(\frac{d|T|}{n}+\alpha\sqrt{\frac{|T|}{i}}\bigg)+\sum_{i=\delta n+1}^n\frac{p(i)}{d}\cdot \bigg(\frac{d|T|}{n}+\alpha\sqrt{\frac{|T|}{\delta n}}\bigg)
\\
&= \frac{|T|}{n}\cdot \sum^{n}_{i=1}p(i) 
  + \bar\alpha\sqrt{|T|}\cdot \sum_{i=1}^{\delta n}\frac{p(i)}{\sqrt i}
  + \frac{\bar\alpha\sqrt{|T|}}{\sqrt{\delta n}} \cdot \sum^n_{i=\delta n+1}p(i)
\\
&\leq \frac{|T|}{n}+\frac{\bar\alpha\sqrt{|T|}}{\sqrt{\delta n}}
  + \bar\alpha\sqrt{|T|}\cdot \sum_{i=1}^{n}\frac{p(i)}{\sqrt i},
\end{align*}
where the last inequality follows from $\sum^{n}_{i=1}p(i)=1$. 
Recall that $Pp(\sigma(1))\geq \cdots \geq Pp(\sigma(n))$ and $|T|=k$.
Hence, by an averaging argument,
\begin{equation}\label{e:Pp(sigma(k))}
Pp(\sigma(k))
\leq \Big(\frac{1}{n}+\frac{1}{\sqrt{k}}\cdot \frac{\bar\alpha}{\sqrt{\delta n}}\Big)
  +  \frac{\bar\alpha}{\sqrt k}\cdot \sum_{i=1}^{n}\frac{p(i)}{\sqrt i}.
\end{equation}
For the case that $k>\delta n$, we simply use the sortedness of $Pp$ to obtain that
\begin{equation}\label{e:Pp(sigma(delta n))}
Pp(\sigma(k))
\leq Pp(\sigma(\delta n))
\leq \Big(\frac{1}{n}+\frac{\bar\alpha}{\delta n}\Big)
  + \frac{\bar\alpha}{\sqrt {\delta n}}\cdot\sum_{i=1}^{n}\frac{p(i)}{\sqrt i}.
\end{equation}
Let $L_1(k):=\frac{1}{n}+\frac{1}{\sqrt{k}}\cdot \frac{\bar\alpha}{\sqrt{\delta n}}$ 
and $H_1(k):=\frac{\bar\alpha}{\sqrt k}\cdot \sum_{i=1}^{n}\frac{p(i)}{\sqrt i}$ denote the two terms in \eqref{e:Pp(sigma(k))}.
Similarly, let $L_2:=\frac{1}{n}+\frac{\bar\alpha}{\delta n}$ 
and $H_2:=\frac{\bar\alpha}{\sqrt {\delta n}}\cdot \sum_{i=1}^{n}\frac{p(i)}{\sqrt i}$ denote the two terms in \eqref{e:Pp(sigma(delta n))}. 
It follows that
\begin{align*}
\|Pp\|_2^2&=\sum^{\delta n}_{k=1}Pp(\sigma(k))^2+\sum^{n}_{k=\delta n+1}Pp(\sigma(k))^2
\\
&\leq \sum^{\delta n}_{k=1}\Big(L_1(k)^2+2L_1(k)H_1(k)+H_1(k)^2\Big)
  + \sum_{k=\delta n+1}^n\Big(L_2^2+2L_2H_2+H_2^2\Big)
\\
&\leq \sum^{\delta n}_{k=1}L_1(k)^2+\sum^n_{k=\delta n+1}L_2^2+\sum^n_{k=1}\Big(2L_2H_2+H_2^2\Big)
\\
&=\sum^{\delta n}_{k=1}L_1(k)^2 +(1-\delta)n\cdot L_2^2+n\cdot \big(2L_2H_2+H_2^2\big),
\end{align*}
where the second inequality follows since $L_1(k)\leq L_2$ and $H_1(k)\leq H_2$ for all $k\leq \delta n$. 

Let $\gamma:=\bar\alpha/\delta$. 
It remains to analyze the three terms. 
For the first term,
\begin{equation}
\begin{split}
\sum^{\delta n}_{k=1}L_1(k)^2
=\sum^{\delta n}_{k=1}\Big(\frac{1}{n}+\frac{1}{\sqrt{k}}\cdot \frac{\bar\alpha}{\sqrt{\delta n}}\Big)^2
&\leq \sum^{\delta n}_{k=1}\frac{1}{n^2}+\frac{2\bar\alpha}{n\sqrt{\delta n}}\sum^{\delta n}_{k=1}\frac{1}{\sqrt k}+\frac{\bar\alpha^2}{\delta n}\sum^n_{k=1}\frac1k
\\
&\leq \frac{1}{n}\cdot \Big(\delta+O\big(\bar\alpha+\bar\alpha\gamma\log n\big)\Big)\\
&\leq \frac1n\cdot \Big(\delta+O\Big(\frac{\bar\alpha\log d}{\xi} \Big)\Big),
\end{split}
\end{equation}
where the second inequality uses $\sum_{k=1}^{\delta n} 1/\sqrt{k} = O(\sqrt{\delta n})$ and 
$\sum_{k=1}^n 1/k = O(\log n)$,
and the last inequality follows by our assumptions that $\gamma\leq O(1)$ and $n=d^{1/\xi}$. 

For the second term, using the assumption that $\gamma\leq \frac{\sqrt 5}{2}-1$,
\begin{equation}
(1-\delta)n\cdot L_2^2
=(1-\delta)n\cdot \Big(\frac1n+\frac{\bar\alpha}{\delta n}\Big)^2
=(1-\delta)\cdot \frac{(1+\gamma)^2}{n}
\leq (1-\delta)\cdot \frac{5}{4n}.
\end{equation}
For the last term, applying the Cauchy-Schwarz inequality, 
\begin{align*}
2L_2H_2+H_2^2
&=\frac{2(1+\gamma)}{n}\cdot \frac{\bar\alpha}{\sqrt{\delta n}}\sum^{n}_{i=1}\frac{p(i)}{\sqrt i}+\frac{\bar\alpha^2}{\delta n}\Big(\sum^{n}_{i=1}\frac{p(i)}{\sqrt i}\Big)^2
\\
&\leq \frac{2\bar\alpha(1+\gamma)}{n\sqrt{\delta n}}\cdot \sqrt{\sum^n_{i=1}\frac1i}\cdot  \sqrt{\sum^n_{i=1}p(i)^2}+\frac{\bar\alpha^2}{\delta n}\cdot \Big(\sum^n_{i=1}\frac1i\Big)\cdot \Big(\sum^n_{i=1}p(i)^2\Big)
\\
&\lesssim\frac{2\bar\alpha(1+\gamma)}{n\sqrt{\delta n}}\cdot \sqrt{\log n}\cdot  \|p\|_2+\frac{\bar\alpha^2}{\delta n}\cdot \log n\cdot \|p\|_2^2.
\end{align*}
Since $n=d^{1/\xi}$ and $\|p\|_2 \geq 1/\sqrt{n}$, it follows that
\begin{align*}
n\cdot \big(2L_2H_2+H_2^2\big)&\lesssim \frac{2\bar\alpha(1+\gamma)}{\sqrt{\delta}}\cdot \sqrt{\frac{\log d}{\xi} }\cdot  \|p\|^2_2+\frac{\bar\alpha^2}{\delta}\cdot \frac{\log d}{\xi} \cdot \|p\|_2^2
\\
&\lesssim \sqrt \gamma(1+\gamma)\cdot \sqrt{\frac{\bar\alpha\log d}{\xi}}+\gamma\cdot \frac{\bar\alpha\log d}{\xi} \cdot \|p\|^2_2&\tag{by $\gamma=\frac{\bar\alpha}{\delta}$}
\\
&\lesssim \sqrt{\frac{\bar\alpha\log d}{\xi}} \cdot \|p\|_2^2,
\end{align*}
where the last inequality uses the assumptions that $\gamma\leq O(1)$ 
and $(\bar\alpha\log d)/\xi\leq 1$. 
 
Combining the three terms and using $1/n\leq \|p\|^2_2$ and $(\bar\alpha\log d)/\xi\leq 1$,
we conclude that
\begin{align*}
\|Pp\|_2^2
&\leq \frac1n\cdot \Big(\delta+O\Big(\frac{\bar\alpha\log d}{\xi}\Big)\Big) + (1-\delta)\cdot \frac{5}{4n}+O\bigg(\sqrt{\frac{\bar\alpha\log d}{\xi}}\bigg)\cdot \|p\|_2^2
\\
&=\frac1n\cdot \Big(\delta+(1-\delta)\cdot \frac54\Big) + O\Big(\frac{\bar\alpha\log d}{\xi}\Big)\cdot \frac1n+O\bigg(\sqrt{\frac{\bar\alpha\log d}{\xi}}\bigg)\cdot \|p\|_2^2
\\
&\leq \frac1n\cdot \Big(\delta+(1-\delta)\cdot \frac54\Big)+O\bigg(\sqrt{\frac{\bar\alpha\log d}{\xi}}\bigg)\cdot \|p\|_2^2.
\end{align*}
\end{proof}

\subsection{Proof of \autoref{t:VD}}

\autoref{p:2-Norm Bound} allows us to derive the number of steps needed for the squared 2-norm of any initial distribution to drop to around $C_\delta/n$.

\begin{lemma}\label{c:steps to C_delta/n}
Let $G$ be a $d$-regular graph with $d=n^\xi$ for some constant $\xi>0$.
If $G$ satisfies the $\delta$-small-set $\alpha$-bipartite density condition for some $\alpha = d^{1-c} \leq \xi^2 d / (\log d)^2$ and $\delta\gtrsim \alpha /d$, then 
\[
\|P^tp_0\|_2^2\leq \frac{C_\delta}{n}+O\Big(\frac{1}{n^{1+c\xi/4}}\Big)
\quad \textrm{for~any~initial~distribution~} p_0 \textrm{~and~any~} 
t\geq \frac{4}{c\xi}+1.
\]
\end{lemma}

\begin{proof}
Let $1/\beta := O\big(\sqrt{\alpha\log d/(\xi d)}\big)$ denote the drop rate in \autoref{p:2-Norm Bound}. 
For any probability distribution $p$, if $\|p\|_2^2\lesssim C_\delta\beta/n$, then by \autoref{p:2-Norm Bound},
\[
\|Pp\|_2^2
\leq \frac{C_\delta}{n}+\frac1\beta\cdot \|p\|_2^2
\leq \frac{C_\delta}{n}+\frac1\beta\cdot O\Big(\frac{C_\delta \beta}{n}\Big)
\leq O\Big(\frac{C_\delta}{n}\Big).
\]
This implies that after the next step of random walks,
\[
\|P^2p\|_2^2
\leq \frac{C_\delta}{n}+\frac1\beta\cdot \|Pp\|_2^2
\leq\frac{C_\delta}{n}+\frac1\beta\cdot O\Big(\frac{C_\delta}{n}\Big)
= \frac{C_\delta}{n}+O\bigg(\sqrt{\frac{\alpha\log d}{\xi d}}\cdot \frac{C_\delta}{n}\bigg)
= \frac{C_\delta}{n}+O\bigg(\sqrt{\frac{\log d}{ \xi d^{c}}}\cdot \frac{1}{n}\bigg).
\]
Note that our assumption $\alpha = d^{1-c} \leq \xi^2 d / (\log d)^2$ implies that $d^{\frac{c}{2}} \geq (\log d) / \xi$, and thus
\[
\|P^2p\|_2^2
\leq \frac{C_\delta}{n}+O\bigg(\sqrt{\frac{\log d}{ \xi d^{c}}}\cdot \frac{1}{n}\bigg)
\leq \frac{C_\delta}{n}+O\Big(\frac{1}{d^{c/4}}\cdot \frac{1}{n}\Big)
= \frac{C_\delta}{n}+O\Big(\frac{1}{n^{c\xi/4 + 1}}\Big),
\]
where the last equality follows by the assumption $d=n^\xi$. 
Hence, for distribution $p$ that is already close to stationary distribution  
(i.e., $\|p\|^2\lesssim C_\delta \beta/n$), 
its squared 2-norm drops to $C_\delta/n+O(1/n^{c\xi/4 + 1})$ in two steps. 

On the other hand, if $\|p\|^2\geq C_\delta\beta/n$, there is a large drop rate such that
\[
\|Pp\|_2^2\leq \frac{C_\delta}{n}+\frac1\beta\cdot \|p\|_2^2\leq \frac2\beta\cdot \|p\|_2^2.
\]
This implies that for any initial distribution $p_0$, $\|P^tp_0\|_2^2\lesssim C_\delta\beta/n$ for $t\geq 4/c\xi-1$ because
\[
\|P^tp_0\|_2^2
\leq \Big(\frac2\beta\Big)^t\cdot \|p_0\|_2^2
\leq \Big(\frac{2}{\beta}\Big)^{\frac{4}{c\xi}-1}
=\frac\beta2\cdot O\bigg( \sqrt{\frac{\log d}{\xi d^{c}}} \bigg)^{\frac{4}{c\xi}}
\lesssim \frac\beta2\cdot \Big(\frac{1}{d^{c/4}}\Big)^{\frac{4}{c\xi}}
\leq \frac{\beta}{n}
\leq \frac{C_\delta\beta}{n},
\]
where the third last inequality uses $d^{\frac{c}{2}} \geq (\log d) / \xi$ that we derived above. 

To summarize, $\|P^tp_0\|_2^2\lesssim C_\delta\beta/n$ after $4/c\xi-1$ steps of random walks, and the lemma follows after two more steps of random walks using the calculation in the first paragraph.
\end{proof}

We are ready to prove \autoref{t:VD}. 

\begin{proof}[Proof of \autoref{t:VD}]
Let $p$ be an arbitrary probability distribution. 
Note that an upper bound on $\|p\|_2^2$ implies an upper bound on $d_{TV}(p, \pi)$ by Cauchy-Schwarz:
\begin{equation}\label{dTV and 2norm}
d_{\TV}(p, \pi)
=\frac12\|p-\pi\|_1
\leq \frac{\sqrt n}{2}\|p-\pi\|_2
=\frac{\sqrt n}{2}\sqrt{\|p\|_2^2-2\langle p,\pi\rangle+\|\pi\|_2^2}
=\frac{\sqrt n}{2}\sqrt{\|p\|_2^2-\frac1n}.
\end{equation}
By \autoref{c:steps to C_delta/n}, for any initial distribution $p_0$ and any $t\geq 4/c\xi+1$,\[
\|P^tp_0\|^2_2\leq \frac{C_\delta}{n} + O\Big(\frac{1}{n^{1+c\xi/4}}\Big).
\] 
Combining this with \eqref{dTV and 2norm} and the fact that $\sqrt{a+b}\leq \sqrt a+\sqrt b$, we conclude that
\begin{align*}
d_{\TV}\big(P^tp_0,\pi\big)
\leq \frac{\sqrt n}{2} \sqrt{\frac{C_\delta-1}{n}+O\Big(\frac{1}{n^{1+c\xi/4}}\Big)}
\leq \frac{\sqrt n}{2} \bigg(\sqrt{\frac{C_\delta-1}{n}}+O\Big(\frac{1}{n^{1/2+c\xi/8}}\Big)\bigg)
=\frac{\sqrt{1-\delta}}{4}+O\Big(\frac{1}{n^{c\xi/8}}\Big)
\end{align*}
where the last equality follows since $C_\delta-1=5(1-\delta)/4 -(1-\delta)=(1-\delta)/4$. 
\end{proof}

\section{Lower Bounds on Mixing Time}\label{s:lower-bound}

We prove \autoref{t:lower-bound} in this section. 
The proof relies on the following lemma on variation distance.

\begin{lemma}\label{l:VD lower bound}
Let $G$ be a $d$-regular graph. 
Suppose there exist $S,T\subseteq V$ such that $|E(S,T)|\geq d|S||T|/n+\alpha\sqrt{|S||T|}$ for some $\alpha$. Then, for any $t\in\mathbb N$,
\[
d_{\TV}(t)\geq \frac12\cdot \Big(\frac{\alpha}{2d}\Big)^{2t}-\frac{\min\{|S|,|T|\}}{2n}.
\]
\end{lemma}

\begin{proof}
We argue that if $S,T$ form a dense bipartite structure, then the random walks starting at $U_S := \chi_S / |S|$ (the uniform distribution on $S$) should bounce back and forth between $S,T$, causing slow mixing time.

Assume without loss of generality that $S,T$ are minimal set that satisfies $|E(S,T)|\geq d|S||T|/n+\alpha\sqrt{|S||T|}$, and that $|S|=\min\{|S|,|T|\}$. 
First, we argue that a dense bipartite structure implies a lower bound on the degree for each $v\in S$ and $u\in T$. 
For any $v\in S$, by minimality,
\begin{align*}
|E(S,T)|-d_T(v)=|E(S\setminus \{v\}, T)|&<\frac{d(|S|-1)|T|}{n}+\alpha\sqrt{(|S|-1)|T|}
\\
&<\frac{d|S||T|}{n}+\alpha\sqrt{|S||T|}-\alpha\sqrt{|T|}(\sqrt{|S|}-\sqrt{|S|-1})
\\
&\leq |E(S,T)|-\frac\alpha2\sqrt{\frac{|T|}{|S|}},   
\end{align*}
where the last inequality uses $\sqrt{x}-\sqrt{x-1}>\frac{1}{2\sqrt x}$.
This implies a lower bound on $d_T(v)$ and similarly on $d_S(u)$ for $u\in T$ such that
\[
d_T(v)\geq \frac\alpha2\sqrt{\frac{|T|}{|S|}}
\quad\text{and}\quad 
d_S(v)\geq \frac\alpha2\sqrt{\frac{|S|}{|T|}}.
\]
Let $d_{\min} := \min\Big\{\frac\alpha2\sqrt{\frac{|T|}{|S|}}, \frac\alpha2\sqrt{\frac{|S|}{|T|}}\Big\}$. 
Since $d_{\min}\leq d$, it also follows that 
\[
|S|\geq \frac{\alpha^2}{4d^2}|T|
\quad\text{and}\quad 
|T|\geq \frac{\alpha^2}{4d^2}|S|
\quad\implies \quad 
d_{\min}\geq \frac{\alpha^2}{4d}.
\]
Consider the random walk starting at $U_S$.  
After one step of random walks, since $d_S(u)\geq d_{\min}$ for all $u\in T$, $PU_S(u)\geq d_{\min}/(d|S|)$. 
Similarly, for any $v\in S$, after another step of random walks, 
$P^2U_S(v)\geq d_{\min}^2/(d^2|S|)$. 
By induction, after the $t$-th step, for all $v\in S$,
\[
P^tU_S(v)
\geq \Big(\frac{d_{\min}}{d}\Big)^t\cdot \frac{1}{|S|}
\geq \Big(\frac{\alpha^2}{4d^2}\Big)^t\cdot \frac{1}{|S|}.
\]
It follows from the definition of $d_{\TV}(t)$ that
\[
d_{\TV}(t)
\geq \frac12 \|P^tU_S-\pi\|_1
\geq \frac12\sum_{v\in S}\Big|P^tU_S(v)-\frac1n\Big|
\geq \frac12 \cdot |S|\cdot \Big(\Big(\frac{\alpha^2}{4d^2}\Big)^t\cdot \frac{1}{|S|}-\frac1n\Big)
=\frac12\cdot \Big(\frac{\alpha}{2d}\Big)^{2t}-\frac{|S|}{2n}.
\]  
\end{proof}

\autoref{t:lower-bound} follows immediately from \autoref{l:VD lower bound}.

\begin{proof}[Proof of \autoref{t:lower-bound}]
It follows from \autoref{l:VD lower bound} and the assumption $|S|,|T|\leq \delta n$ that for any $t\in\mathbb N$, 
\[
d_{\TV}(t)
\geq \frac12\cdot \Big(\frac{\alpha}{2d}\Big)^{2t}-\frac{\min\{|S|,|T|\}}{2n}
\geq \frac12\cdot \Big(\frac{\alpha}{2d}\Big)^{2t}-\frac{\delta}{2}
\]
For $d_{\TV}(t)$ to be below $\frac1n$, we need
\[
\frac12\cdot \Big(\frac{\alpha}{2d}\Big)^{2t}-\frac{\delta}{2}\leq \frac1n
\quad\implies\quad 
t\geq \frac{\log(\delta+\frac2n)}{2\log \alpha/(2d)}\gtrsim\frac{\log\delta}{\log (\alpha/d)} 
= \frac{\log(1/\delta)}{\log(d/\alpha)}.
\]
\end{proof}

\subsection{Expanders and Small-Set Vertex Expanders} \label{s:SSE}

Some natural combinatorial conditions to consider for constant mixing time of a graph are its expansion properties, such as edge conductance and small-set vertex expansion. 

We discuss why they are not strong enough to attain constant mixing time in this section. 
The graphs that we construct have a small dense bipartite structure embedded while having near optimal edge conductance or small-set vertex expansion. 
It follows from \autoref{t:lower-bound} that it does not have constant mixing time.

\subsubsection*{Counterexample for Expanders}

\textbf{Example:} Let $G$ be a $d$-regular graph with $d\geq 4$, where there exist two small disjoint sets $S,T\subseteq V$ such that $|S|=|T|=n/(d+2) \leq n/d$.
Each vertex in $S$ has $d/2$ edges into $T$, and similarly each vertex in $T$ has $d/2$ edges into $S$.
The remaining $d/2$ edges of each $v\in S\cup T$ go to unique neighbours in $V\setminus S\cup T$. 
The induced subgraph $G[V\setminus S\cup T]$ forms a graph of degree $d-1$ with edge conductance $1/2$.

\begin{claim}
    $G$ has edge conductance at least $1/8$.
\end{claim}

\begin{proof}
Let $W\subseteq V$. Let $W_1=W\cap (S\cup T)$, and let $W_2=W\cap (V\setminus S\cap T)$. 
By construction, $|\delta(W_1)|\geq \frac d2|W_1|$ and $|\delta(W_2)|\geq \frac{d-1}{2} |W_2|$. 
Assume $|W_1|\geq |W_2|$. Since $|E(W_1,W_2)|\leq |W_2|$ by construction, it follows that
\[
|\delta(W)|\geq |\delta(W_1)|-|W_2|\geq \frac d2 \cdot |W_1|-|W_1|\geq \frac d4 \cdot |W_1|\geq \frac d8 \cdot |W|,
\]
where the second last inequality holds for any $d\geq 4$. 
Assume $|W_2|\geq |W_1|$, then
\[
|\delta(W)|\geq |\delta_{V\setminus S\cup T}(W_2)|\geq \frac{d-1}{2}\cdot |W_2|\geq \frac d4\cdot |W_1|\geq \frac d8\cdot |W|.
\]
We conclude that $G$ has edge conductance at least $1/8$.
\end{proof}

It remains to see that $G$ has a small dense bipartite structure.

\begin{claim}
    $$|E(S,T)|\geq \frac{d|S||T|}{n}+\frac d4\sqrt{|S||T|}$$
\end{claim}

\begin{proof}
    Since $|S|=|T|\leq n/d$, we have
    $$\frac{d|S||T|}{n}+\frac d4\sqrt{|S||T|}\leq |S|+\frac d4|S|\leq \frac d2|S|=|E(S,T)|$$
    where the last inequality follows for any $d\geq 4$.
\end{proof}

It follows from \autoref{t:lower-bound} that $G$ does not have constant mixing time.

\subsubsection*{Counterexample for Small-Set Vertex Expanders}

To see that near-optimal small set vertex expansion also fails to attain constant mixing time, consider the same graph construction, except now the induced subgraph $G[V\setminus S\cup T]$ is small-set vertex expander of size $dn/(d+2)$ and degree $d-1$, 
where each set of size at most $dn/((d+2)(d-1))$ has vertex expansion at least $d/2$. 
The small dense bipartite structure $(S,T)$ still exists in $G$, 
so $G$ does not have constant mixing time by the same argument above. 
It suffices to check that $G$ has near optimal small-set vertex expansion.

\begin{claim}
For every subset $W$ of size at most $n/d$, its vertex expansion in $G$ is at least $d/8$.
\end{claim}
\begin{proof}
Let $W\subseteq V$ with $|W|\leq n/d$, define $N(W)=\{u\in V\mid\exists v\in W\text{ with } uv\in E\}$ as the neighbor set of $W$. We prove the claim by lower bounding $N(W)$. Let $W_1=W\cap (S\cup T)$ and $W_2=W\cap (V\setminus S\cap T)$.
By construction, $|N(W_1)|\geq d|W_1|/2$ and $|N(W_2)|\geq d|W_2|/2$. 
Assume $|W_1|\geq |W_2|$. 
Since $|N(W_1)|\geq d|W_1|/2$ and $2|W_1|\geq |W_1|+|W_2|=|W|$, 
\[
|N(W)|\geq |N(W_1)|\geq \frac d4\cdot 2|W_1|\geq \dfrac d4 \cdot |W|.
\]
The case where $|W_2|\geq |W_1|$ follows similarly. 

Therefore, we can lower bound the vertex expansion as
$$|\partial(W)|\geq |N(W)|-|W|\geq (\frac d4-1)|W|\geq \frac d8\cdot |W|$$
where the last inequality holds for any $d\geq 8$.
\end{proof}

\subsection*{Acknowledgement}
We thank Thatchaphol Saranurak for discussions on length-constrained expanders that inspired the study in this work.

\bibliography{reference}

\end{document}